\documentclass[11pt]{amsart}
\usepackage{fullpage,graphicx,mathpazo,color}
\usepackage{amsmath,amscd,amssymb}
\usepackage{mathrsfs}
\usepackage{subfigure}
\usepackage{epstopdf}
\usepackage{amsmath}
\usepackage{tikz}
\allowdisplaybreaks[4]
\usepackage[colorlinks, citecolor=blue,pagebackref,hypertexnames=false]{hyperref}

\begin{document}

\newtheorem{theorem}{Theorem}[section]
\newtheorem{lemma}{Lemma}[section]
\newtheorem{proposition}{Proposition}[section]
\newtheorem{corollary}{Corollary}[section]
\newtheorem{remark}{Remark}[section]
\numberwithin{equation}{section}

\newcommand{\bfR}{{\Bbb R}}
\newcommand{\bfC}{{\Bbb C}}
\newcommand{\bfZ}{{\Bbb Z}}
\newcommand{\ii}{\text{i}}
\newcommand{\e}{\text{e}}
\newcommand{\dd}{\text{d}}
\newcommand{\Om}{\omega}
\newcommand{\nn}{\nonumber}
\newcommand\be{\begin{equation}}
\newcommand\ee{\end{equation}}
\newcommand{\bea}{\begin{eqnarray}}
\newcommand{\eea}{\end{eqnarray}}
\newcommand\berr{\begin{eqnarray*}}
\newcommand\eerr{\end{eqnarray*}}

\newcommand{\change}[1]{{\color{blue}#1}}
\newcommand{\note}[1]{{\color{red}#1}}

\title{The Coupled Hirota Equation with a $3\times3$ Lax pair: Painlev\'e-type Asymptotics in Transition Zone}

\author{Xiaodan Zhao}
\address{ School of Control and Computer Engineering, North China Electric Power University, Beijing 102206, P.R. China}
 \email{xiaodan\_zhao0224@163.com (X. Zhao)}

\author{Lei Wang$^*$}
\address{School of Mathematics and Physics, North China Electric Power University, Beijing 102206, P.R. China}
\email{50901924@ncepu.edu.cn (L. Wang)}


\begin{abstract}
We consider the Painlev\'e asymptotics for a solution of integrable coupled Hirota equation with a $3\times3$ Lax pair whose initial data decay rapidly at infinity. Using Riemann--Hilbert techniques and Deift--Zhou nonlinear steepest descent arguments, in a transition zone defined by $|x/t-1/(12\alpha)|t^{2/3}\leq C$, where $C>0$ is a constant, it turns out that the leading-order term to the solution can be expressed in terms of the solution of a coupled Painlev\'e II equation associated with a $3\times3$ matrix Riemann--Hilbert problem.

\end{abstract}
\date{}

\thanks{$^*$Corresponding author.}
\maketitle

\section{Introduction}

The aim of this paper is to probe into the asymptotics as $t\rightarrow\infty$ of a solution to the Cauchy problem for an integrable coupled Hirota equation \cite{PSM-PRE,RLD-JPA,TP-JMP},
\begin{align}
&\left\{
\begin{aligned}
&\ii u_t+\frac{1}{2}u_{xx}+(|u|^2+|v|^2)u+\ii\alpha[u_{xxx}+(6|u|^2+3|v|^2)u_x+3uv^*v_x]=0,\\
&\ii v_t+\frac{1}{2}v_{xx}+(|u|^2+|v|^2)v+\ii\alpha[v_{xxx}+(6|v|^2+3|u|^2)v_x+3vu^*u_x]=0,
\end{aligned}
\right.\label{1.1}\\
&u(x,0)=u_0(x), \ v(x,0)=v_0(x), \ x\in\bfR,\label{1.2}
\end{align}
where $u(x,t)$, $v(x,t)$ are the complex envelop functions, $\alpha>0$ is a small real parameter denoting the strength of higher-order effects, $u_0(x)$ and $v_0(x)$ are sufficiently smooth and rapidly decaying as $|x|\rightarrow\infty$. The terms inside square brackets in Equation \eqref{1.1} account for the third-order dispersion, self-steepening, and delayed nonlinear response effect, respectively. For illustration the transmission procedure when high intensity ultra-short pulses traverse an optical glass fiber \cite{BMP}, they are turned out to be non-negligible in optics. Thus, compared with the simple Manakov system (namely, $\alpha=0$), the coupled Hirota system can be considered to be  a more accurate prototype of the wave evolution in the real world.

Mathematically, thanks to the integrability, a series of important results have been obtained on the coupled Hirota equation. Early in 1992, Tasgal and Potasek shown that Equation \eqref{1.1} can be formulated in terms of an eigenvalue problem, and thus is solvable by the means of inverse scattering transformation \cite{TP-JMP}. In \cite{RLD-JPA}, the bright and dark $N$-soliton solutions have been obtained by using Hirota bilinearization derivable from the Painlev\'e-analysis. In 2006, Huo and Jia established the local well-posedness of the Cauchy problem for the coupled Hirota equation in Sobolev spaces $H^s\times H^s$ ($s\geq1/4$) by the Fourier restriction norm method \cite{HJ-JMAA}. In addition, various rogue wave solutions of the coupled Hirota equation \eqref{1.1}, such as lowest-order fundamental, dark and composite rogue waves as well as dark-bright-rogue waves and rogue-wave pairs   were derived in \cite{C,CS-PRE,WC-PLB}. The interactional solutions between rogue waves and the other nonlinear waves such as breathers and dark-bright solitons of \eqref{1.1} have been reported in \cite{WLC}.
In 2019, Zhang et al. studied the modulational instability, multi-dark soliton and higher-order vector rogue wave   structures of the coupled Hirota equation via the Darboux transformation \cite{ZYW}.

The study on long-time asymptotic behaviors is an important and challengeable topic in integrable system. A powerful tool for analyzing the large-time asymptotics of the integrable nonlinear evolution equations is the nonlinear steepest descent method, which was firstly proposed by Deift and Zhou \cite{DZ1993}. Since then this new method has been widely applied to find numerous new significant asymptotic results in a rigorous form for different nonlinear integrable models \cite{AL-Non,CL-JLMS,HLJL,LG-JDE,WF-CMP,XJF}. In particular, for the long-time asymptotic analysis of solutions to coupled Hirota equation \eqref{1.1}, the initial-boundary value problem on the half-line has been first considered in \cite{LG-SCM} based on the nonlinear steepest descent analysis for an appropriate Riemann--Hilbert (RH) problem characterization. Moreover, for the Schwartz initial values $u(x,0), \ v(x,0)\in\mathcal{S}(\bfR)$, the asymptotic behavior as $t\to\infty$ for Cauchy problem of the coupled Hirota equation \eqref{1.1} was carried out by Liu and the first author \cite{LZ}. It turns out that there are two asymptotic sectors (a slowly decaying of the order $O(t^{-1/2})$, and a sector of rapid decay) with a transition region between the adjacent sectors in $(x,t)$-half-plane, see Figure \ref{fig1}. More precisely, the two sectors are defined as follows: (I) $x/t<1/(12\alpha)-\varepsilon$ (oscillatory sector), (II) $x/t>1/(12\alpha)+\varepsilon$ (sector of rapid decay ), where $\varepsilon>0$. However, as far as we know, how to describe the asymptotics in the transition region, $x/t\approx1/(12\alpha)$, still lies on the table. Thus, the primary target of our paper is to derive the Painlev\'e-type asymptotics of the solution to the Cauchy problem \eqref{1.1}-\eqref{1.2} in the transition region defined by
\begin{align}\label{1.3}
\mathcal{T}\doteq\left\{(x,t)\in\bfR\times\bfR^+:
\left|\frac{x}{t}-\frac{1}{12\alpha}\right|t^{\frac{2}{3}}\leq C\right\},
\end{align}
where $C>0$ is a constant.
\begin{figure}[htbp]
  \centering
  \includegraphics[width=4in]{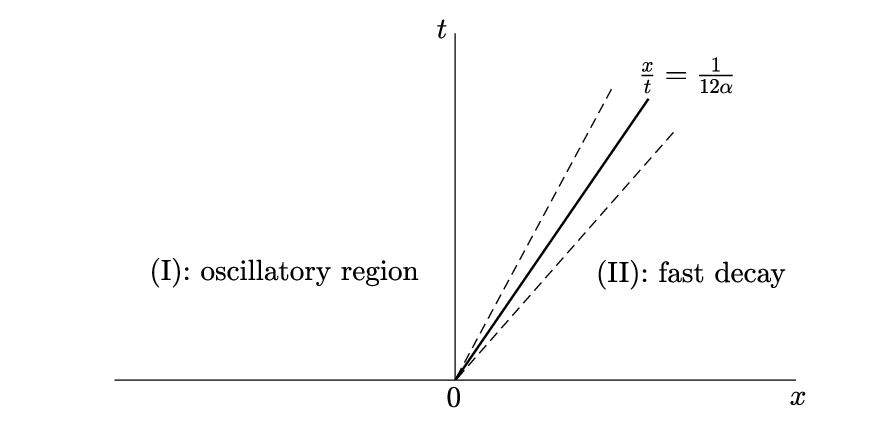}
  \caption{The different sectors of the $(x,t)$-half-plane.}\label{fig1}
\end{figure}

In the literature, we note that the Painlev\'e transcendents and their higher order analogues are crucial in asymptotic analysis of many integrable nonlinear partial differential equations. The Painlev\'e asymptotics was first appeared in the case of the Korteweg--de Vries (KdV) equation by Segur and Ablowitz \cite{SA-PD} and the modified KdV equation by Deift and Zhou \cite{DZ1993} both in the self-similar sector $|x|\leq Ct^{1/3}$. In 2010, Boutet de Monvel, Its, and Shepelsky computed the long-time asymptotics of the solution of Camassa--Holm equation in two transition regions and shown that in both regions the asymptotics is expressed in terms of second Painlev\'e transcendents \cite{AID}. Moreover, the Painlev\'e-type asymptotics was also discovered for an extended modified KdV equation and defocusing Hirota in transition sector \cite{LG-JDE,XJF}. The higher order Airy and Painlev\'e asymptotics for the modified KdV equation and modified KdV hierarchy were considered in \cite{CL-JLMS} and \cite{HZ-SIAM}, respectively. Interestingly, the Painlev\'e III hierarchy has arisen in study of the fundamental rogue wave solutions of the focusing nonlinear Schr\"odinger equation in the limit of large order \cite{BLP}. It is noted that for the solutions of Sasa--Satsuma equation and matrix mKdV equation associated with higher order matrix spectral problems, there also exist Painlev\'e-type asymptotics \cite{HLJL,LZG}. Recently, considering the defocusing NLS equation with a nonzero background, Wang and Fan found the Painlev\'e asymptotics in two transition regions \cite{WF-CMP}. The main result is now stated as follows.

\begin{theorem}\label{Main-Theorem}
Let $u_0(x)$ and $v_0(x)$ be two functions in the Schwartz class and let $\varrho(k)=\begin{pmatrix}\varrho_1(k)& \varrho_2(k)\end{pmatrix}$ be the reflection coefficient defined by \eqref{2.6}. Assume that the $(1,1)$ entry $s_{11}(k)$ of $s(k)$ described in \eqref{2.4} is nonzero for Im$k\leq0$. Then, for any constant $C>0$, the solution of the Cauchy problem \eqref{1.1}-\eqref{1.2} to the coupled Hirota equation satisfies the following long-time asymptotic formula for $(x,t)\in\mathcal{T}$, i.e., $|x/t-1/(12\alpha)|t^{2/3}\leq C$,
\begin{align}\label{1.4}
u(x,t)=\frac{2\e^{\ii\psi_1(y,t)}}{(3\alpha t)^{\frac{1}{3}}}u_p(y)+O(t^{-\frac{2}{3}}),\quad
v(x,t)=\frac{2\e^{\ii\psi_2(y,t)}}{(3\alpha t)^{\frac{1}{3}}}v_p(y)+O(t^{-\frac{2}{3}}),
\end{align}
where
\begin{align}
y=&(3\alpha)^{-\frac{1}{3}}\left(\frac{x}{t}-\frac{1}{12\alpha}\right)t^{\frac{2}{3}},\nn\\
\psi_j(y,t)=&\arg \varrho_j\left(-\frac{2}{3}\right)+\frac{1}{216\alpha^2}t+
\frac{1}{2}(3\alpha)^{-\frac{2}{3}}t^{\frac{1}{3}}y,\quad j=1,2,\nn
\end{align}
and $u_p(y)$, $v_p(y)$ denote the smooth solution of the coupled Painlev\'e
II equation \eqref{A.5} corresponding to $s_1=|\varrho_1(-2/3)|$, $s_2=|\varrho_2(-2/3)|$ according to Theorem \ref{lemmaA.1}.
\end{theorem}

\begin{remark}
Theorem \ref{Main-Theorem} extends the result of the Painlev\'e asymptotics for the defocusing Hirota equation with $2\times2$ matrix Lax pair in \cite{XJF} to the coupled Hirota equation \eqref{1.1} associated with $3\times3$ matrix spectral problem. In the case of the Hirota equation considered in \cite{XJF}, the corresponding final approximation RH problem can be solved directly by the solution of the standard homogeneous Painlev\'e II equation. However, the key ingredient of our analysis consists of proposing a new $3\times3$ matrix model RH problem and establishing the connection with the Painlev\'e II function (Theorem \ref{lemmaA.1}), which are the main differences between our proof and that in \cite{XJF}. Moreover, in order to ensure the solvability of model RH problem, we need to assume that the constants $s_1$ and $s_2$ are real in the jump matrices \eqref{A.1}. Nevertheless, the entries at relevant positions of the jump matrix in our final $3\times3$ matrix approximation RH problem are complex. With the aid of a suitable gauge transformation, the solution of approximation problem is expressed in terms of the solution of new model RH problem.
\end{remark}

The structure of this paper can now be explained. In Section \ref{sec2}, we quickly review the presentation of a basic RH formalism $M(x,t;k)$ related to the Cauchy problem \eqref{1.1}-\eqref{1.2} for the coupled Hirota equation. For more details, see \cite{LZ}. The asymptotics in the transition region will be discussed in Section \ref{sec3}. A substantial part of the work, namely, the constructions of appropriate local models, is deferred to the two Appendices \ref{secA} and \ref{secb}.

\section{The basic Riemann--Hilbert problem}\label{sec2}

The system \eqref{1.1} is integrable and admits the following $3\times3$  Lax pair representation \cite{WLC}
\begin{equation}\label{2.1}
\left\{
\begin{aligned}
&\psi_x(x,t;k)=\frac{\ii}{12\alpha}k\Sigma\psi(x,t;k)+U(x,t)\psi(x,t;k),\\
&\psi_t(x,t;k)=\frac{\ii}{192\alpha^2}(k^3+2k^2)\Sigma\psi(x,t;k)+V(x,t;k)\psi(x,t;k),\\
\end{aligned}
\right.
\end{equation}
where $\psi(x,t;k)$ is a $3\times3$ matrix-valued spectral function, $k$ is a complex iso-spectral parameter, $\Sigma$, $U(x,t)$ and $V(x,t;k)$ are given by
\begin{align}
\Sigma=&\begin{pmatrix}
-2 & 0 & 0\\
0 & 1 & 0\\
0 & 0 & 1\end{pmatrix},\quad U=\begin{pmatrix}
0 & -u & -v\\
u^* & 0 & 0\\
v^* & 0 & 0
\end{pmatrix},\\
V=&\frac{k^2}{16\alpha}U+k\left(\frac{\ii}{4}\Sigma_1(U^2-U_x)+\frac{1}{8\alpha}U\right)\\
&+\frac{\ii}{2}\Sigma_1(U^2-U_x)+\alpha\left([U_x,U]-U_{xx}+2U^3\right),\quad \Sigma_1=\text{diag}(-1,1,1).\nn
\end{align}

Given $u_0(x),\ v_0(x)\in\mathcal{S}(\bfR)$, we can define the scattering matrix $s(k)$ by
\be\label{2.4}
s(k)=\mathbb{I}_{3\times3}-\int_{-\infty}^{+\infty}\e^{-\frac{\ii}{12\alpha}k x\Sigma}[U(x,0)X(x;k)]\e^{\frac{\ii}{12\alpha}k x\Sigma}\dd x,\quad k\in\bfR,
\ee
where the $3\times3$ matrix-valued function $X(x;k)$ is the unique solution of the Volterra integral equation
\be
X(x;k)=\mathbb{I}_{3\times3}-\int^{+\infty}_x\e^{\frac{\ii}{12\alpha}k(x-y)
\Sigma}[U(y,0)X(y;k)]\e^{-\frac{\ii}{12\alpha}k(x-y)\Sigma}\dd y,\quad x\in\bfR,\ k\in\bfR.
\ee
Then the ``reflection coefficient" is defined by
\be\label{2.6}
\varrho(k)\doteq\begin{pmatrix}\varrho_1(k)& \varrho_2(k)\end{pmatrix}
=\begin{pmatrix}\frac{s_{12}(k)}{s_{11}(k)}&\frac{s_{13}(k)}{s_{11}(k)}\end{pmatrix},
\quad k\in\bfR.
\ee
It can be shown that the $(1,1)$ entry $s_{11}(k)$ of $s(k)$ is  analytic in the lower half-plane. Possible zeros of $s_{11}(k)$ give rise to poles in the RH problem. For simplicity, we assume that there is no such pole (solitonless case). Then the main RH problem associated with the Cauchy problem \eqref{1.1}-\eqref{1.2} is as follows \cite{LZ}.

\begin{theorem}\label{th2.1}
Define the $3\times3$ matrix-valued function $J(x,t;k)$ by
\be
J(x,t;k)=\begin{pmatrix}
1+\varrho(k)\varrho^\dag(k) & \varrho(k)\e^{-t\Phi(k)}\\
\varrho^\dag(k)\e^{t\Phi(k)} & \mathbb{I}_{2\times2}
\end{pmatrix},
\ee
where ``$\dagger$" denotes the a matrix transpose conjugate and
\be\label{2.8}
\Phi(k)=\frac{\ii}{4\alpha}\left[\frac{x}{t}k+\frac{1}{16\alpha}(k^3+2k^2)\right].
\ee
Then the following $3\times3$ matrix RH problem:\\
$\bullet$ Analyticity: $M(x,t;k)$ is analytic for $k\in\bfC\setminus\bfR$ and is continuous for $k\in\bfR$;\\
$\bullet$ Jump relation: the boundary values  $M_{\pm}(x,t;k)=\lim_{\varepsilon\to0}M(x,t;k\pm\ii\varepsilon)$  satisfy the jump condition $M_+(x,t;k)=M_-(x,t;k)J(x,t;k)$ for $k\in\bfR$;\\
$\bullet$ Normalization: $M(x,t;k)\to\mathbb{I}_{3\times3}$ as $k\rightarrow\infty$;\\
has a unique solution for each $(x,t)\in\bfR\times\bfR^+$. Moreover, the functions $u(x,t)$ and $v(x,t)$ defined by
\be\label{2.9}
u(x,t)=-\frac{\ii}{4\alpha}\lim_{k\rightarrow\infty}(kM(x,t;k))_{12},\quad
v(x,t)=-\frac{\ii}{4\alpha}\lim_{k\rightarrow\infty}(kM(x,t;k))_{13}
\ee
is a smooth function with rapid decay as $|x|\rightarrow\infty$ which satisfies the coupled Hirota equation \eqref{1.1} for $(x,t)\in\bfR\times\bfR^+$. Furthermore, $u(x,0)=u_0(x)$, $v(x,0)=v_0(x)$.
\end{theorem}

\begin{figure}[htbp]
\centering
    \subfigure[ ]
    {\label{fig2-1}
        \begin{minipage}[t]{0.5\textwidth}
           \centering
            \includegraphics[width=1.0\textwidth]{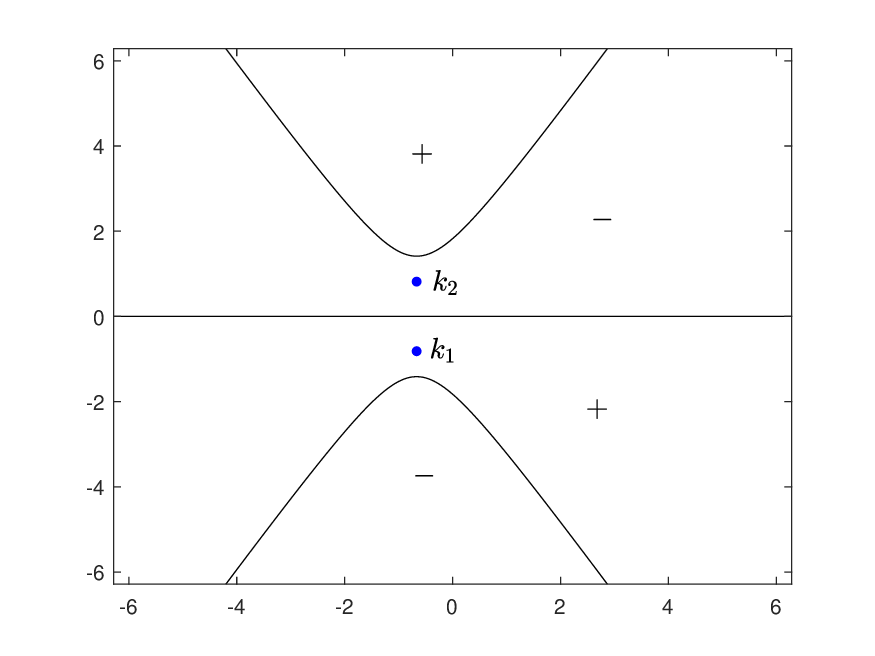}
        \end{minipage}
    }
    \subfigure[ ]
    {\label{fig2-2}
        \begin{minipage}[t]{0.5\textwidth}
        \centering
            \includegraphics[width=1.0\textwidth]{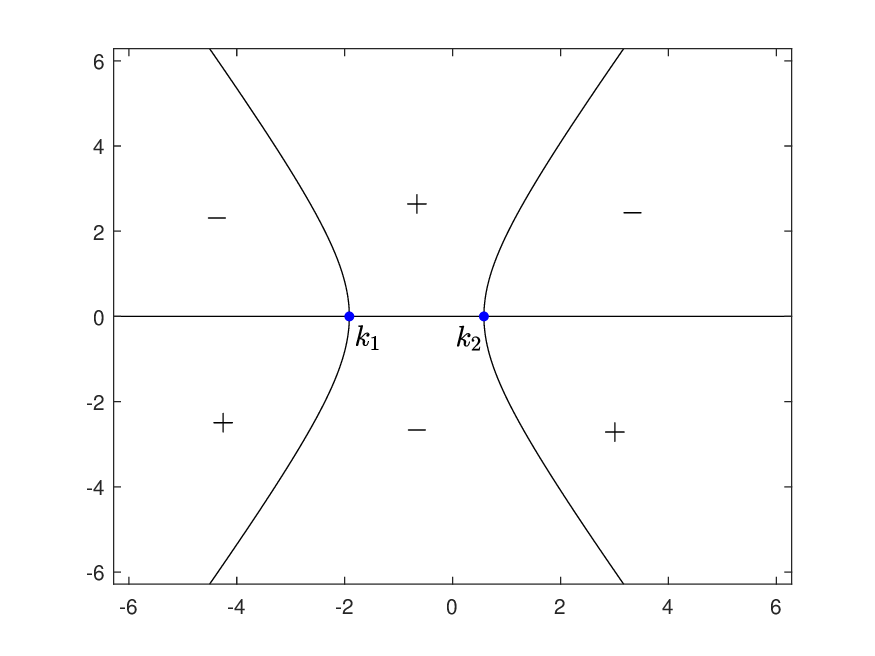}
        \end{minipage}
    }
\caption{Signature tables for $0<(\frac{x}{t}-\frac{1}{12\alpha})t^{\frac{2}{3}}<C$ and $-C<(\frac{x}{t}-\frac{1}{12\alpha})t^{\frac{2}{3}}<0$. The signs $\pm$ correspond to domains with $\pm$Re$\Phi(k)>0$.}\label{fig2}
\end{figure}

The nonlinear steepest descent approach for oscillatory RH problems consists in a battery of
 transformations of the problem, so as to arrive at a model RH problem that can be solved explicitly. The considerable contour transformation is  determined by the signature table relative to the region (in the $(x,t)$-half-plane) of interest, i.e., the distribution of the signs of Re$\Phi(k)$. For the transition zone, the signature tables are shown in Figure \ref{fig2}, where
\begin{align}
k_1&=\frac{2}{3}\left(-1-\sqrt{1-12\alpha\frac{x}{t}}\right),\label{2.10}\\
k_2&=\frac{2}{3}\left(-1+\sqrt{1-12\alpha\frac{x}{t}}\right)\label{2.11}.
\end{align}

\section{Asymptotics in transition zone: $\left|\frac{x}{t}-\frac{1}{12\alpha}\right|t^{\frac{2}{3}}<C$}\label{sec3}

In this section, our goal is to find the asymptotics of solution to Cauchy problem \eqref{1.1}-\eqref{1.2} in the transition region $\mathcal{T}$ defined by \eqref{1.3}. We will derive the asymptotics in the two halves of Sector $\mathcal{T}$ corresponding to $x/t\geq1/(12\alpha)$ and $x/t\leq1/(12\alpha)$ separately; we thus will use the notation
\be
\mathcal{T}_{\geq}\doteq\mathcal{T}\cap\left\{\frac{x}{t}\geq\frac{1}{12\alpha}\right\},\quad
\mathcal{T}_{\leq}\doteq\mathcal{T}\cap\left\{\frac{x}{t}\leq\frac{1}{12\alpha}\right\}.
\ee

\subsection{Asymptotics in Sector $\mathcal{T}_{\geq}$}\label{sub3.1}

Suppose $(x,t)\in\mathcal{T}_{\geq}$, which corresponds to the case in Figure \ref{fig2-1}. In this region, two critical points $k_1$ and $-k_2$ defined in \eqref{2.10} and \eqref{2.11} are complex and approach to $-2/3$ at least the speed
of $t^{-1/3}$ as $t\rightarrow\infty$.

\subsubsection{Modification to the basic RH problem}
We will make the transformation to the basic RH problem in such a way that the new jump matrix approaches the identity matrix as $t\to\infty$.

Let $\Gamma\subset\bfC$ denote the contour $\Gamma=\bfR\cup\{-\frac{2}{3}+\e^{\frac{\pi\ii}{6}}\bfR\}\cup\{-\frac{2}{3}+\e^{-\frac{\pi\ii}{6}}\bfR\}$ oriented  to the right and let $\Omega$ and $\Omega^*$ denote the triangular domains shown in Figure \ref{fig3}. We then decompose $\varrho=\begin{pmatrix}\varrho_1& \varrho_2\end{pmatrix}$ into an analytic part $\varrho_a$ and a small remainder $\varrho_r$.

\begin{figure}[htbp]
\centering
\includegraphics[width=4in]{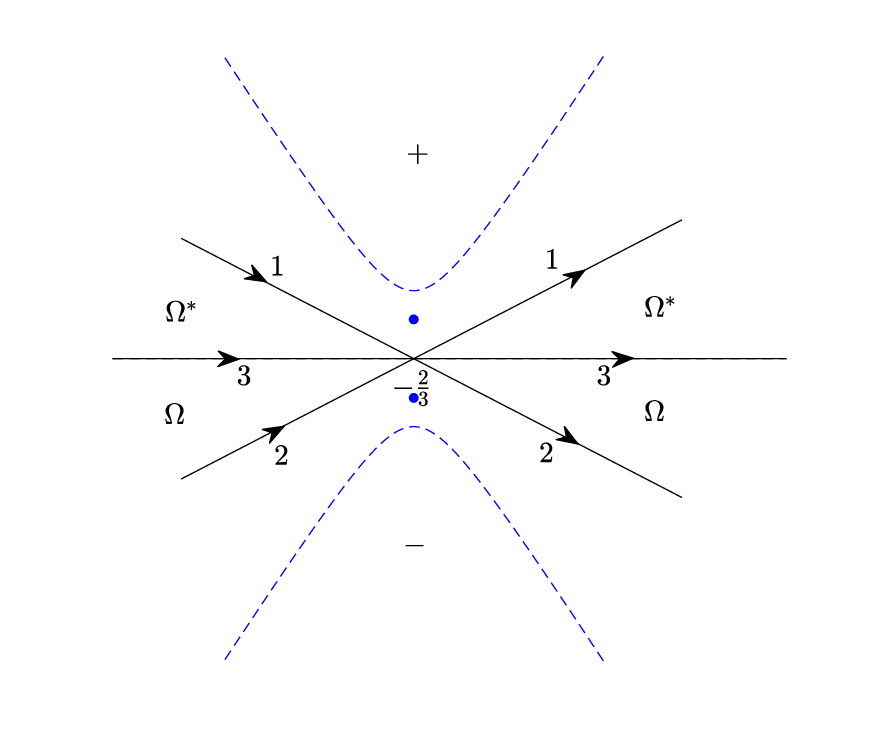}
\caption{The oriented contour $\Gamma$ and sets $\Omega$, $\Omega^*$ .}\label{fig3}
\end{figure}
\begin{lemma}(Analytic approximation)\label{lemma3.1}
There exists a decomposition
\be
\varrho_1(k)=\varrho_{1,a}(t;k)+\varrho_{1,r}(t;k), \quad k\in\bfR,
\ee
where the functions $\varrho_{1,a}$ and $\varrho_{1,r}$ satisfy the following properties:\\
(i) For each $(x,t)\in\mathcal{T}_\geq$, $\varrho_{1,a}(t;k)$ is defined and continuous for $k\in\bar{\Omega}$ and analytic for $k\in\Omega$.\\
(ii) The function $\varrho_{1,a}$ satisfies
\be\label{3.3}
\left\{
\begin{aligned}
&|\varrho_{1,a}(t;k)|\leq \frac{C}{1+|k|^2}\e^{\frac{t}{4}|\text{Re}\Phi(k)|},\\
&\left|\varrho_{1,a}(t;k)-\varrho_1(-\frac{2}{3})\right|\leq C\left|k+\frac{2}{3}\right|\e^{\frac{t}{4}|\text{Re}\Phi(k)|},
\end{aligned}
\right.\quad k\in\bar{\Omega}.
\ee
(iii) The $L^1, L^2$ and $L^\infty$ norms of the function $\varrho_{1,r}(t;\cdot)$ on $\bfR$ are $O(t^{-3/2})$ as $t\rightarrow\infty$.
\end{lemma}
\begin{proof}
See \cite{AL-Non}, Lemma 5.2 or \cite{CL-JLMS}, Lemma 4.1.
\end{proof}
The decomposition of $\varrho_2(k)=\varrho_{2,a}(t;k)+\varrho_{2,r}(t;k)$ can be similarly found. Thus, we obtain a decomposition $\varrho=\varrho_a+\varrho_r$ by letting
\berr
\varrho_a(t;k)\doteq\begin{pmatrix}\varrho_{1,a}(t;k)& \varrho_{2,a}(t;k)\end{pmatrix},\quad
\varrho_r(t;k)\doteq\begin{pmatrix}\varrho_{1,r}(t;k)& \varrho_{2,r}(t;k)\end{pmatrix}.
\eerr

Now we can deform the contour by setting
\be\label{3.4}
M^{(1)}(x,t;k)=\e^{-\frac{1}{3}t\Phi(-\frac{2}{3})\Sigma}
M(x,t;k)G(x,t;k)\e^{\frac{1}{3}t\Phi(-\frac{2}{3})\Sigma},
\ee
where
\be\label{3.5}
G(x,t;k)=\left\{
\begin{aligned}
&\begin{pmatrix}
1 & \varrho_a(t;k)\e^{-t\Phi(k)}\\
\mathbf{0}_{2\times1} & \mathbb{I}_{2\times2}
\end{pmatrix},\qquad k\in \Omega,\\
&\begin{pmatrix}
1 & \mathbf{0}_{1\times2}\\
-\varrho_a^\dag(t;k^*)\e^{t\Phi(k)} & \mathbb{I}_{2\times2}
\end{pmatrix},\quad\,\ k\in \Omega^*,\\
&\mathbb{I}_{3\times3},\qquad\qquad\qquad\qquad\qquad\,\,\, \text{elsewhere}.
\end{aligned}
\right.
\ee
It follows that $M^{(1)}(x,t;k)$ satisfies the $3\times3$ matrix RH problem:\\
$\bullet$ $M^{(1)}(x,t;k)$ is analytic off $\Gamma$, and it takes continuous boundary values on $\Gamma$;\\
$\bullet$ Across the oriented contour $\Gamma$, the boundary values $M^{(1)}_\pm(x,t;k)$ are connected by the relation $M^{(1)}_+(x,t;k)=M^{(1)}_-(x,t;k)J^{(1)}(x,t;k);$\\
$\bullet$ $M^{(1)}(x,t;k)\to\mathbb{I}_{3\times3}$, as $k\to\infty$;\\
where the jump matrix $J^{(1)}(x,t;k)$ is given by
\begin{align}
J^{(1)}_1=&\begin{pmatrix}
1 & \mathbf{0}_{1\times2}\\
\varrho_a^\dag\e^{t(\Phi(k)-\Phi(-\frac{2}{3}))} & \mathbb{I}_{2\times2}
\end{pmatrix},\quad J^{(1)}_2=\begin{pmatrix}
1 & \varrho_a\e^{-t(\Phi(k)-\Phi(-\frac{2}{3}))}\\
\mathbf{0}_{2\times1} & \mathbb{I}_{2\times2}
\end{pmatrix},\nn\\
J^{(1)}_3=&\begin{pmatrix}
1 & \varrho_r\e^{-t(\Phi(k)-\Phi(-\frac{2}{3}))}\\
\mathbf{0}_{2\times1} & \mathbb{I}_{2\times2}
\end{pmatrix}\begin{pmatrix}
1 & \mathbf{0}_{1\times2}\\
\varrho_r^\dag\e^{t(\Phi(k)-\Phi(-\frac{2}{3}))} & \mathbb{I}_{2\times2}
\end{pmatrix},
\end{align}
where $J^{(1)}_j$ denotes the restriction of $J^{(1)}$ to the subcontour labeled by $j$ in Figure \ref{fig3}.

\subsubsection{Local model}

An important observation is that
\be\label{3.7}
t\left(\Phi(k)-\Phi(-\frac{2}{3})\right)=2\ii\left(yz+\frac{4}{3}z^3\right),
\ee
where $z$ is the scaled spectral parameter
\be\label{3.8}
z\doteq\frac{1}{8}\left(\frac{3}{\alpha^2}\right)^{\frac{1}{3}}
\left(k+\frac{2}{3}\right)t^{\frac{1}{3}},
\ee
and
\be\label{3.9}
y\doteq(3\alpha)^{-\frac{1}{3}}\left(\frac{x}{t}-\frac{1}{12\alpha}\right)t^{\frac{2}{3}}.
\ee
Fix suitable $\epsilon>0$ and let $D_\epsilon(-2/3)=\{k\in\bfC||k+2/3|<\epsilon\}$, see Figure \ref{fig4}. Let $\Gamma^\epsilon=(\Gamma\cap D_\epsilon(-2/3))\setminus\bfR$. Then the map $k\mapsto z$ takes $\Gamma^\epsilon$ onto $L^\epsilon\doteq L\cap\{|z|<1/8(3/\alpha^2)^{1/3}t^{1/3}\epsilon\}$, where $L$ is the contour defined in \eqref{A.0}.
Thus for fixed $z$ and large $t$, the jump matrix $J^{(1)}(x,t;k)$ can be approximated as follows:
\be\label{3.10}
J^{(1)}(x,t;k)\rightarrow\left\{
\begin{aligned}
&\begin{pmatrix}
1 & \mathbf{0}_{1\times2}\\
\varrho^\dag\left(-\frac{2}{3}\right)\e^{2\ii(yz+\frac{4}{3}z^3)} & \mathbb{I}_{2\times2}
\end{pmatrix}, \qquad z\in L^\epsilon_1,\\
&\begin{pmatrix}
1 & \varrho\left(-\frac{2}{3}\right)\e^{-2\ii(yz+\frac{4}{3}z^3)}\\
\mathbf{0}_{2\times1} & \mathbb{I}_{2\times2}
\end{pmatrix},\qquad z\in L^\epsilon_2.
\end{aligned}
\right.
\ee

Now, we write $\varrho_j(-2/3)=|\varrho_j(-2/3)|\e^{\ii\phi_j}$ with $\phi_j=\arg \varrho_j(-2/3)$ for $j=1,2$. We conclude that as $t\rightarrow\infty$, $M^{(1)}(x,t;k)$ in $D_\epsilon(-\frac{2}{3})$ approaches the solution $M^c(x,t;k)$ defined by
\be\label{3.11}
M^c(x,t;k)\doteq\e^{\frac{\ii\phi_1}{2}\mathcal{A}}\e^{\frac{\ii\phi_2}{2}\mathcal{B}}M^L(y;z)
\e^{-\frac{\ii\phi_2}{2}\mathcal{B}}\e^{-\frac{\ii\phi_1}{2}\mathcal{A}},
\ee
where
\be\label{3.12}
\mathcal{A}=\begin{pmatrix}1 & 0 &0\\0&-1&0\\0&0&1\end{pmatrix},\quad
\mathcal{B}=\begin{pmatrix}1 & 0 &0\\0&1&0\\0&0&-1\end{pmatrix},
\ee
and $M(y;z)$ is the solution of the coupled Painlev\'e II RH problem of Theorem \ref{lemmaA.1} with $s_1\doteq |\varrho_1(-2/3)|$, $s_2\doteq |\varrho_2(-2/3)|$.
\begin{figure}[htbp]
\centering
\includegraphics[width=4in]{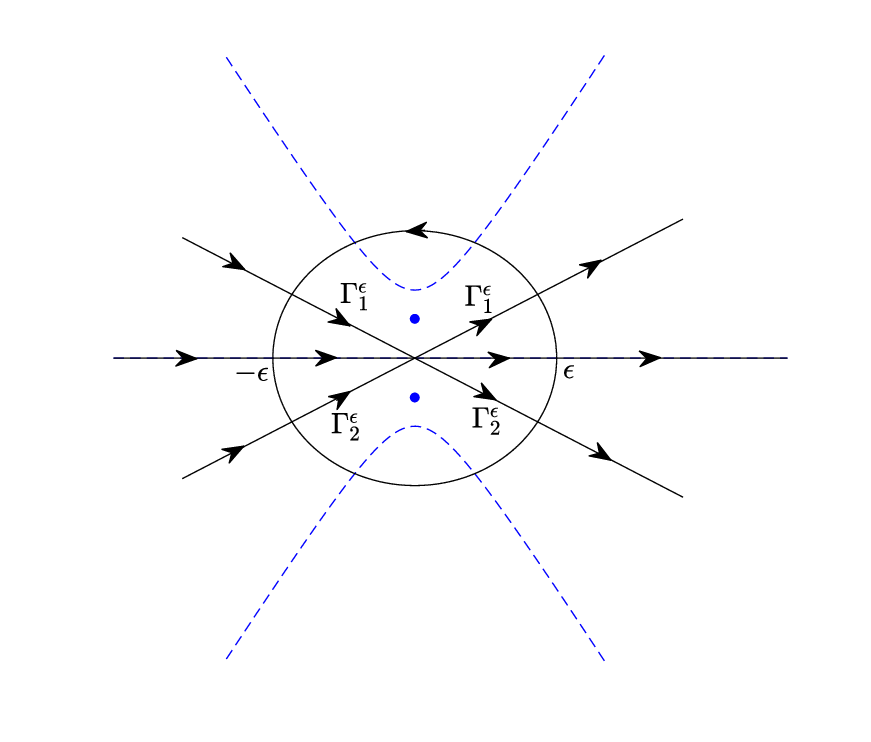}
\caption{The contour $\Gamma^\epsilon=\cup_{j=1}^2\Gamma^\epsilon_j$ and $\hat{\Gamma}$.}\label{fig4}
\end{figure}
\begin{lemma}\label{lemma3.2}
For each $(x,t)\in\mathcal{T}_\geq$, the function $M^c(x,t;k)$ defined by \eqref{3.11} is an analytic function of $k\in D_\epsilon(-\frac{2}{3})\setminus\Gamma^\epsilon$ such that $|M^c(x,t;k)|\leq C$. Across $\Gamma^\epsilon$, $M^c(x,t;k)$ obeys the jump condition $M^c_+=M^c_-J^c$, where the jump matrix $J^c$ satisfies for each $1\leq p\leq\infty$,
\be\label{3.13}
\|J^{(1)}-J^c\|_{L^p(\Gamma^\epsilon)}\leq Ct^{-\frac{1}{3}(1+\frac{1}{p})}.
\ee
In particular, as $t\rightarrow\infty$,
\be\label{3.14}
\|[M^c(x,t;k)]^{-1}-\mathbb{I}_{3\times3}\|_{L^\infty(\partial D_\epsilon(-\frac{2}{3}))}=O(t^{-\frac{1}{3}}),
\ee
and
\be\label{3.15}
\frac{1}{2\pi\ii}\int_{\partial D_\epsilon(-\frac{2}{3})}\left([M^c(x,t;k)]^{-1}-\mathbb{I}_{3\times3}\right)\dd k=-\frac{8\alpha^{\frac{2}{3}}M^c_1(y)}{(3t)^{\frac{1}{3}}}+O(t^{-\frac{2}{3}}),
\ee
where $M^c_1(y)$ satisfies
\be
\left[M^c_1(y)\right]_{12}=\ii\e^{\ii\phi_1} u_p(y),\quad
\left[M^c_1(y)\right]_{13}=\ii\e^{\ii\phi_2} v_p(y),
\ee
moreover, $u_p(y)$ and $v_p(y)$ are the smooth solution of the coupled Painlev\'e II equation \eqref{A.5}.
\end{lemma}
\begin{proof}
The analyticity and boundedness of $M^c$ can be deduced from relating results about Theorem \ref{lemmaA.1}. Moreover, we have
\be\label{3.17}
J^{(1)}-J^{c}=\left\{
\begin{aligned}
&\begin{pmatrix}0 & \mathbf{0}_{1\times2}\\
\left(\varrho_a^\dag(t;k^*)-\varrho^\dag\left(-\frac{2}{3}\right)\right)\e^{t(\Phi(k)-\Phi(-\frac{2}{3}))} & \mathbf{0}_{2\times2}\end{pmatrix},\quad k\in\Gamma^\epsilon_1,\\
&\begin{pmatrix}0 & \left(\varrho_a(t;k)-\varrho\left(-\frac{2}{3}\right)\right)\e^{-t(\Phi(k)-\Phi(-\frac{2}{3}))}\\
\mathbf{0}_{2\times1} & \mathbf{0}_{2\times2}\end{pmatrix},\quad\,\ k\in\Gamma^\epsilon_2.
\end{aligned}
\right.
\ee
In the following proof of \eqref{3.13}, we take the analysis of the case $k+2/3=l\e^{\pi\ii/6}$ with $l\leq0$ as an example. For $x/t\geq1/(12\alpha)$, since
\be
\text{Re}\left(\Phi(k)-\Phi(-\frac{2}{3})\right)=-\frac{l^3}{64\alpha^2}
-\frac{l}{8\alpha}\left(\frac{x}{t}-\frac{1}{12\alpha}\right)
\geq-\frac{l^3}{64\alpha^2}\geq\frac{|k+\frac{2}{3}|^3}{64\alpha^2},
\ee
and hence,
\be
\e^{-\frac{3}{4}t|\text{Re}(\Phi(k)-\Phi(-\frac{2}{3}))|}\leq\e^{-2|z|^3}.
\ee
It follows from \eqref{3.17} and \eqref{3.3} that
\begin{align}\label{3.20}
|J^{(1)}-J^{c}|\leq&\left|\varrho_a(t;k)-\varrho\left(-\frac{2}{3}\right)\right|
\e^{-t\text{Re}(\Phi(k)-\Phi(-\frac{2}{3}))}\nn\\
\leq& C\left|zt^{-\frac{1}{3}}\right|\e^{-\frac{3}{4}t|\text{Re}(\Phi(k)-\Phi(-\frac{2}{3}))|}
\leq C\left|zt^{-\frac{1}{3}}\right|\e^{-2|z|^3}.
\end{align}
Consequently, writing $l=|z|$, we get
\be
\|J^{(1)}-J^c\|_{L^\infty}\leq C\sup_{0\leq l<\infty}\left(lt^{-\frac{1}{3}}\right)\e^{-2l^3}\leq Ct^{-\frac{1}{3}},
\ee
and
\be
\|J^{(1)}-J^c\|_{L^1}\leq C\int_0^\infty l t^{-\frac{1}{3}}\e^{-2l^3}\frac{\dd l}{t^{\frac{1}{3}}}\leq Ct^{-\frac{2}{3}}.
\ee
Therefore, \eqref{3.13} follows from the general inequality $\|f\|_{L^p}\leq\|f\|^{1-1/p}_{L^\infty}\|f\|_{L^1}^{1/p}$.

If $k\in\partial D_\epsilon(-2/3)$, the variable $z$ tends to infinity as $t\rightarrow\infty$. It follows from the asymptotic expansion \eqref{A.3} and Cauchy's formula that \eqref{3.14} and \eqref{3.15} hold.
\end{proof}

\subsubsection{The final step}

Set $\hat{\Gamma}=\Gamma\cup\partial D_\epsilon(-2/3)$ and assume that the boundary of
$D_\epsilon(-2/3)$ are oriented counterclockwise, see Figure \ref{fig4}. Then the function $\hat{M}(x,t;k)$ defined by
\be\label{3.23}
\hat{M}(x,t;k)=\left\{
\begin{aligned}
&M^{(1)}(x,t;k)[M^c(x,t;k)]^{-1},\quad k\in D_\epsilon\left(-\frac{2}{3}\right),\\
&M^{(1)}(x,t;k),\qquad\qquad\qquad\quad{\text elsewhere}.
\end{aligned}
\right.
\ee
satisfies the following $3\times3$ matrix RH problem:\\
$\bullet$ $\hat{M}(x,t;k)$ is analytic outside the contour $\hat{\Gamma}$ with continuous boundary values on $\hat{\Gamma}$;\\
$\bullet$ For $k\in\hat{\Gamma}$, we have the jump relation
$\hat{M}_+(x,t;k)=\hat{M}_-(x,t;k)\hat{J}(x,t;k);$\\
$\bullet$ $\hat{M}(x,t;k)\to\mathbb{I}_{3\times3}$, as $k\to\infty$;\\
where the jump matrix is given by
\be\label{3.24}
\hat{J}=\left\{
\begin{aligned}
&M^{c}_-J^{(1)}[M^{c}_+]^{-1},\quad k\in\hat{\Gamma}\cap D_\epsilon\left(-\frac{2}{3}\right),\\
&[M^{c}]^{-1},\qquad\quad\quad\,\ k\in\partial D_\epsilon\left(-\frac{2}{3}\right),\\
&J^{(1)},\qquad\qquad\qquad  k\in\hat{\Gamma}\setminus \overline{D_\epsilon\left(-\frac{2}{3}\right)}.
\end{aligned}
\right.
\ee
Let $\hat{w}=\hat{J}-\mathbb{I}_{3\times3}$. It then follows from the proof of Lemma 4.3 in \cite{CL-JLMS} that we have the following conclusion.
\begin{lemma}
For $(x,t)\in\mathcal{T}_{\geq}$ and $1\leq p\leq\infty$,
\begin{align}
\|\hat{w}\|_{L^p(\partial D_\epsilon(-\frac{2}{3}))}\leq& Ct^{-\frac{1}{3}},\label{3.25}\\
\|\hat{w}\|_{L^p(\bfR)}\leq&Ct^{-\frac{3}{2}},\label{3.26}\\
\|\hat{w}\|_{L^p(\Gamma^\epsilon)}\leq&Ct^{-\frac{1}{3}(1+\frac{1}{p})},\label{3.27}\\
\|\hat{w}\|_{L^p(\hat{\Gamma}')}\leq&C\e^{-ct},\label{3.28}
\end{align}
where $\hat{\Gamma}'=\hat{\Gamma}\setminus (\bfR\cup\overline{D_\epsilon(-2/3)})$.
\end{lemma}
We now derive the asymptotic formula of the solution of Cauchy problem \eqref{1.1}-\eqref{1.2} for the coupled Hirota equation in sector $\mathcal{T}_{\geq}$. Let $\hat{C}$ be the Cauchy operator associated with $\hat{\Gamma}$ and let $\hat{C}_{\hat{w}}g\doteq\hat{C}_-(g\hat{w})$. Then by estimates \eqref{3.25}-\eqref{3.28}, we can write
\be\label{3.29}
\hat{M}(x,t;k)=\mathbb{I}_{3\times3}+\frac{1}{2\pi\ii}
\int_{\hat{\Gamma}}\frac{(\hat{\mu}\hat{w})(x,t;s)}{s-k}\dd s,
\ee
where the $3\times3$ matrix-valued function $\hat{\mu}(x,t;k)$ is defined by $\hat{\mu}=\mathbb{I}_{3\times3}+\hat{C}_{\hat{w}}\hat{\mu}.$
Moreover, the Neumann series argument implies that
\be\label{3.30}
\|\hat{\mu}(x,t;\cdot)-\mathbb{I}_{3\times3}\|_{L^2(\hat{\Gamma})}=O(t^{-\frac{1}{3}}),\quad t\rightarrow\infty.
\ee
By expanding \eqref{3.29} and inverting the transformations \eqref{3.4} and \eqref{3.23}, we find the relation,
\be\label{3.31}
\lim_{k\rightarrow\infty}k(M(x,t;k)-\mathbb{I}_{3\times3})
=-\frac{1}{2\pi\ii}\e^{\frac{1}{3}t\Phi(-\frac{2}{3})\Sigma}
\int_{\hat{\Gamma}}(\hat{\mu}\hat{w})(x,t;s)\dd s\e^{-\frac{1}{3}t\Phi(-\frac{2}{3})\Sigma}.
\ee
By \eqref{3.15}, \eqref{3.25} and \eqref{3.30}, on the subcontour $\partial D_\epsilon(-2/3)$, we have
\be
\begin{aligned}
-\frac{1}{2\pi\ii}\int_{\partial D_\epsilon(-\frac{2}{3})}(\hat{\mu}\hat{w})(x,t;s)\dd s
=&-\frac{1}{2\pi\ii}\int_{\partial D_\epsilon(-\frac{2}{3})}\hat{w}\dd s
-\frac{1}{2\pi\ii}\int_{\partial D_\epsilon(-\frac{2}{3})}(\hat{\mu}-\mathbb{I}_{3\times3})\hat{w}\dd s\\
=&\frac{8\alpha^{\frac{2}{3}}M^c_1(y)}{(3t)^{\frac{1}{3}}}+O(t^{-\frac{2}{3}}).
\end{aligned}
\ee
In addition, the contributions from $\bfR$, $\Gamma^\epsilon$ and $\hat{\Gamma}'$ to the right-hand side of \eqref{3.31} are $O(t^{-3/2})$, $O(t^{-2/3})$ and $O(e^{-ct})$, respectively.
Recalling the reconstructional formula \eqref{2.9} and the definition of $\Phi(k)$ in \eqref{2.8}, we immediately obtain the asymptotic formula \eqref{1.4}.

\subsection{Asymptotics in Sector $\mathcal{T}_{\leq}$}\label{sub3.2}
We now analyze the asymptotics in the Sector $\mathcal{T}_{\leq}$. In this sector, the critical points $k_1$ and $k_2$ given respectively in \eqref{2.10} and \eqref{2.11} are real and as $t\to\infty$ approach $-2/3$ at least as fast as $t^{-1/3}$.

\subsubsection{Modification to the basic RH problem}
In this case, we define the oriented contour $\Xi=\bfR\cup\Xi_1\cup\Xi_2$ and open subsets $\Omega$ and $\Omega^*$ as depicted in Figure \ref{fig5}, where
\be
\begin{aligned}
\Xi_1=&\{k=k_2+l\e^{\frac{\pi\ii}{6}}|l\geq0\}\cup\{k=k_1+l\e^{\frac{5\pi\ii}{6}}|l\geq0\},\\
\Xi_2=&\{k=k_2+l\e^{-\frac{\pi\ii}{6}}|l\geq0\}\cup\{k=k_1+l\e^{-\frac{5\pi\ii}{6}}|l\geq0\}.
\end{aligned}
\ee

\begin{figure}[htbp]
\centering
\includegraphics[width=4in]{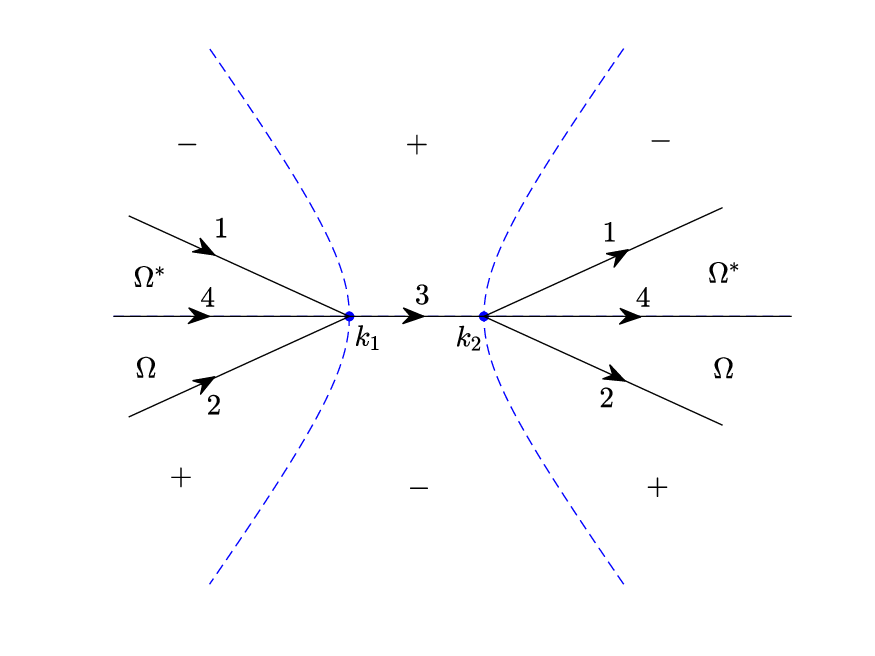}
\caption{The oriented contour $\Xi$ and sets $\Omega$, $\Omega^*$.}\label{fig5}
\end{figure}
\begin{lemma}(Analytic approximation)\label{lemma3.3}
There exists a decomposition
\be
\varrho_1(k)=\varrho_{1,a}(x,t;k)+\varrho_{1,r}(x,t;k), \quad k\in(-\infty,k_1)\cup(k_2,\infty),
\ee
where the functions $\varrho_{1,a}$ and $\varrho_{1,r}$ satisfy the following properties:\\
(i) For $(x,t)\in\mathcal{T}_\leq$, $\varrho_{1,a}(x,t;k)$ is defined and continuous for $k\in\bar{\Omega}$ and analytic for $\Omega$.\\
(ii) The function $\varrho_{1,a}$ satisfies
\be
\left\{
\begin{aligned}
&|\varrho_{1,a}(x,t;k)|\leq \frac{C}{1+|k|^2}\e^{\frac{t}{4}|\text{Re}\Phi(k)|},
\quad k\in\bar{\Omega},\\
&|\varrho_{1,a}(x,t;k)-\varrho(k_j)|\leq C|k-k_j|\e^{\frac{t}{4}|\text{Re}\Phi(k)|},\quad k\in\bar{\Omega},\ j=1,2.
\end{aligned}
\right.
\ee
(iii) The $L^1, L^2$ and $L^\infty$ norms of the function $\varrho_{1,r}(x,t;\cdot)$ on $(-\infty,k_1)\cup(k_2,\infty)$ are $O(t^{-3/2})$ as $t\rightarrow\infty$.\\
\end{lemma}
\begin{proof}
See \cite{JL-IUMJ}, Lemma 4.8.
\end{proof}
Similar decomposition holds for $\varrho_2(k)$, namely, $\varrho_2(k)=\varrho_{2,a}(x,t;k)+\varrho_{2,r}(x,t;k)$. Then a decomposition of $\varrho=\varrho_a+\varrho_r$ is achieved by setting
\berr
\varrho_a(x,t;k)\doteq\begin{pmatrix}\varrho_{1,a}(x,t;k)& \varrho_{2,a}(x,t;k)\end{pmatrix},\quad
\varrho_r(x,t;k)\doteq\begin{pmatrix}\varrho_{1,r}(x,t;k)& \varrho_{2,r}(x,t;k)\end{pmatrix}.
\eerr

Using the decomposition of $\varrho$, we define $M^{(1)}(x,t;k)$ by \eqref{3.4} with $G(x,t;k)$
given by \eqref{3.5}. Hence $M^{(1)}$ satisfies the RH problem:\\
$\bullet$ $M^{(1)}(x,t;k)$ is analytic for $k\in\bfC\setminus\Xi$;\\
$\bullet$ The continuous boundary values $M^{(1)}_\pm$ satisfy
$M^{(1)}_+(x,t;k)=M^{(1)}_-(x,t;k)J^{(1)}(x,t;k)$ for $k\in\Xi$;\\
$\bullet$ $M^{(1)}(x,t;k)\rightarrow \mathbb{I}_{3\times3}$ as $k\rightarrow\infty$;\\
where the jump matrix $J^{(1)}(x,t;k)$ is given by
\begin{align}
J^{(1)}_1=&\begin{pmatrix}
1 & \mathbf{0}_{1\times2}\\
\varrho_a^\dag\e^{t(\Phi(k)-\Phi(-\frac{2}{3}))} & \mathbb{I}_{2\times2}
\end{pmatrix},\quad J^{(1)}_2=\begin{pmatrix}
1 & \varrho_a\e^{-t(\Phi(k)-\Phi(-\frac{2}{3}))}\\
\mathbf{0}_{2\times1} & \mathbb{I}_{2\times2}
\end{pmatrix},\nn\\
J^{(1)}_3=&\begin{pmatrix}
1 & \varrho\e^{-t(\Phi(k)-\Phi(-\frac{2}{3}))}\\
\mathbf{0}_{2\times1} & \mathbb{I}_{2\times2}
\end{pmatrix}\begin{pmatrix}
1 & \mathbf{0}_{1\times2}\\
\varrho^\dag\e^{t(\Phi(k)-\Phi(-\frac{2}{3}))} & \mathbb{I}_{2\times2}
\end{pmatrix},\\
J^{(1)}_4=&\begin{pmatrix}
1 & \varrho_r\e^{-t(\Phi(k)-\Phi(-\frac{2}{3}))}\\
\mathbf{0}_{2\times1} & \mathbb{I}_{2\times2}
\end{pmatrix}\begin{pmatrix}
1 & \mathbf{0}_{1\times2}\\
\varrho_r^\dag\e^{t(\Phi(k)-\Phi(-\frac{2}{3}))} & \mathbb{I}_{2\times2}
\end{pmatrix},\nn
\end{align}
with subscripts referring to Figure \ref{fig5}.

\begin{figure}[htbp]
\centering
\includegraphics[width=4in]{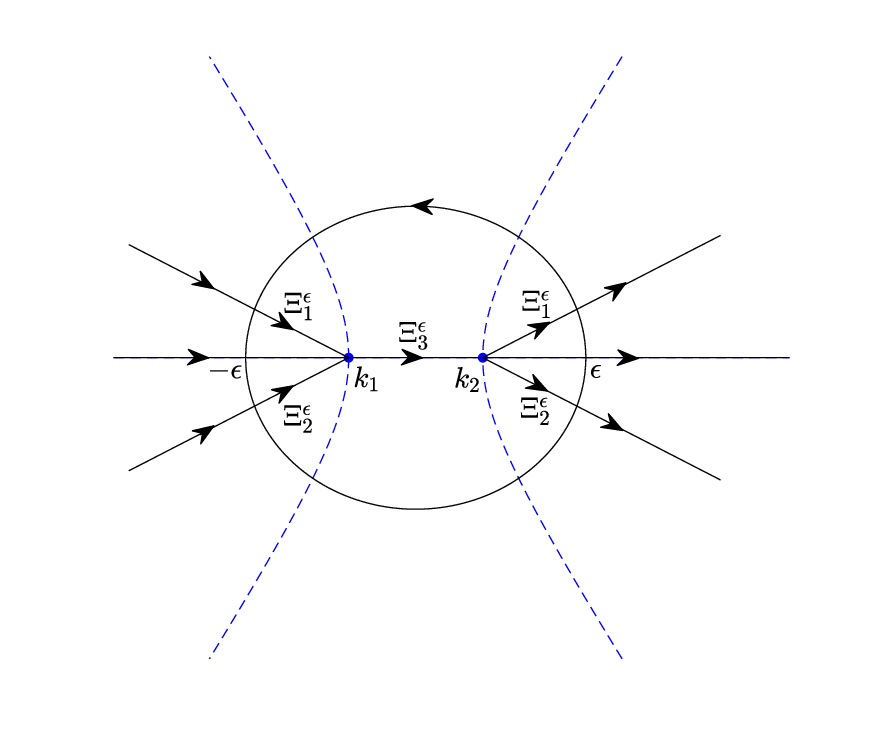}
\caption{The oriented contour $\Xi^\epsilon=\cup_{j=1}^3\Xi^\epsilon_j$ and $\hat{\Xi}$.}\label{fig6}
\end{figure}
\subsubsection{Local model}
As in Subsection \ref{sub3.1}, we introduce the new variables $z$ and $y$ by \eqref{3.8} and \eqref{3.9}, respectively. We now have $-C\leq y\leq0$. Define contour $\Xi^\epsilon=(\Xi\cap D_\epsilon(-2/3))\setminus((-\infty,k_1)\cup(k_2,\infty))$, see Figure \ref{fig6}.
Let $X$ be the contour defined in \eqref{B.1} with $z_2=1/8(3/\alpha^2)^{-1/3}(k_2+2/3)t^{1/3}=\sqrt{-y}/2\geq0.$
The map $k\mapsto z$ maps $\Xi^\epsilon$ onto $X^\epsilon\doteq X\cap\{|z|<1/8(3/\alpha^2)^{1/3}t^{1/3}\epsilon\}$. Therefore, as $t\to\infty,$ we have
\be
J^{(1)}(x,t;k)\rightarrow\left\{
\begin{aligned}
&\begin{pmatrix}
1 & \mathbf{0}_{1\times2}\\
\varrho^\dag\left(-\frac{2}{3}\right)\e^{2\ii(yz+\frac{4}{3}z^3)} & \mathbb{I}_{2\times2}
\end{pmatrix}, \qquad z\in X^\epsilon_1,\\
&\begin{pmatrix}
1 & \varrho\left(-\frac{2}{3}\right)\e^{-2\ii(yz+\frac{4}{3}z^3)}\\
\mathbf{0}_{2\times1} & \mathbb{I}_{2\times2}
\end{pmatrix},\qquad z\in X^\epsilon_2,\\
&\begin{pmatrix}
1 & \varrho\left(-\frac{2}{3}\right)\e^{-2\ii(yz+\frac{4}{3}z^3)}\\
\mathbf{0}_{2\times1} & \mathbb{I}_{2\times2}
\end{pmatrix}\begin{pmatrix}
1 & \mathbf{0}_{1\times2}\\
\varrho^\dag\left(-\frac{2}{3}\right)\e^{2\ii(yz+\frac{4}{3}z^3)} & \mathbb{I}_{2\times2}
\end{pmatrix},\quad z\in X^\epsilon_3.
\end{aligned}
\right.
\ee
The double $(y,z_2)$ belongs to the parameter set $\mathbb{P}$ in \eqref{B.4} whenever $(x,t)\in\mathcal{T}_\leq$. Thus, by Theorem \ref{lemmaB.1}, we can choose
\be\label{3.38}
M^r(x,t;k)\doteq\e^{\frac{\ii\phi_1}{2}\mathcal{A}}\e^{\frac{\ii\phi_2}{2}\mathcal{B}}
M^X(y,z_2;z)\e^{-\frac{\ii\phi_2}{2}\mathcal{B}}\e^{-\frac{\ii\phi_1}{2}\mathcal{A}}
\ee
to approach $M^{(1)}(x,t;k)$ in $D_\epsilon(-2/3)$, where matrices $\mathcal{A},$ $\mathcal{B}$ are given in \eqref{3.12} and $M^X(y,z_2;z)$ is the solution of the model RH problem of Theorem \ref{lemmaB.1} with $z_2=\sqrt{-y}/2$ and $s_1\doteq |\varrho_1(-2/3)|$, $s_2\doteq |\varrho_2(-2/3)|$.

\begin{lemma}\label{lemma3.4}
For each $(x,t)\in\mathcal{T}_\leq$, the function $M^r(x,t;k)$ defined in \eqref{3.38} is analytic for $k\in D_\epsilon(-2/3)\setminus\Xi^\epsilon$ such that $|M^r(x,t;k)|\leq C$. On the contour $\Xi^\epsilon$, $M^r(x,t;k)$ satisfies the jump condition $M^r_+=M^r_-J^r$ and the jump matrix $J^r$ obeys the estimate for each $1\leq p\leq\infty$,
\be\label{3.39}
\|J^{(1)}-J^r\|_{L^p(\Xi^\epsilon)}\leq Ct^{-\frac{1}{3}(1+\frac{1}{p})}.
\ee
In particular,
\begin{align}
\|[M^r(x,t;k)]^{-1}-\mathbb{I}_{3\times3}\|_{L^\infty(\partial D_\epsilon(-\frac{2}{3}))}&=O(t^{-\frac{1}{3}}),\label{3.40}\\
\frac{1}{2\pi\ii}\int_{\partial D_\epsilon(-\frac{2}{3})}
\left([M^r(x,t;k)]^{-1}-\mathbb{I}_{3\times3}\right)\dd k
&=-\frac{8\alpha^{\frac{2}{3}}M^r_1(y)}{(3t)^{\frac{1}{3}}}+O(t^{-\frac{2}{3}}),
\end{align}
where $M^r_1(y)$ satisfies
\be
\left[M^r_1(y)\right]_{12}=\ii\e^{\ii\phi_1} u_p(y),\quad
\left[M^r_1(y)\right]_{13}=\ii\e^{\ii\phi_2} v_p(y).
\ee
\end{lemma}
\begin{proof}
We only give the proof of estimate \eqref{3.39}. First, we know that
\be\label{3.43}
J^{(1)}-J^{r}=\left\{
\begin{aligned}
&\begin{pmatrix}0 & \mathbf{0}_{1\times2}\\
\left(\varrho_a^\dag(x,t;k^*)-\varrho^\dag\left(-\frac{2}{3}\right)\right)\e^{t(\Phi(k)-\Phi(-\frac{2}{3}))} & \mathbf{0}_{2\times2}\end{pmatrix},\quad k\in\Xi^\epsilon_1,\\
&\begin{pmatrix}0 & \left(\varrho_a(x,t;k)-\varrho\left(-\frac{2}{3}\right)\right)\e^{-t(\Phi(k)-\Phi(-\frac{2}{3}))}\\
\mathbf{0}_{2\times1} & \mathbf{0}_{2\times2}\end{pmatrix},\quad\,\ k\in\Xi^\epsilon_2,\\
&\begin{pmatrix}\varrho(k)\varrho^\dag(k)-\varrho\left(-\frac{2}{3}\right)\varrho^\dag\left(-\frac{2}{3}\right) & \left(\varrho(k)-\varrho\left(-\frac{2}{3}\right)\right)\e^{-t(\Phi(k)-\Phi(-\frac{2}{3}))}\\
\left(\varrho^\dag(k)-\varrho^\dag\left(-\frac{2}{3}\right)\right)\e^{t(\Phi(k)-\Phi(-\frac{2}{3}))} & \mathbf{0}_{2\times2}\end{pmatrix},\quad k\in\Xi^\epsilon_3.
\end{aligned}
\right.
\ee
For $k=k_2+l\e^{-\frac{\pi\ii}{6}}$, $0\leq l\leq\epsilon$, we obtain
\be
\text{Re}\left(\Phi(k)-\Phi(-\frac{2}{3})\right)
=\frac{1}{64\alpha^2}\left(l^3+\frac{3\sqrt{3}}{2}(k_2+\frac{2}{3})l^2\right)
\geq\frac{|k-k_2|^3}{64\alpha^2}.
\ee
If $|k-k_2|\geq k_2+2/3$, then $|k-k_2|\geq|k+2/3|/2$, and hence
\be
\e^{-\frac{3}{256\alpha^2}t|k-k_2|^3}\leq\e^{-\frac{1}{1024\alpha^2}t|k+\frac{2}{3}|^3}.
\ee
If $|k-k_2|<k_2+2/3$, then $|k+2/3|\leq2(k_2+2/3)\leq Ct^{-1/3}$, and so
\be
\e^{-\frac{3}{256\alpha^2}t|k-k_2|^3}\leq 1\leq C\e^{-\frac{1}{1024\alpha^2} t|k+\frac{2}{3}|^3}.
\ee
Thus, for $k=k_2+l\e^{\frac{\pi\ii}{6}}$, $0\leq l\leq\varepsilon$, we find
\be
\e^{-\frac{3}{4}t|\text{Re}(\Phi(k)-\Phi(-\frac{2}{3}))|}
\leq\e^{-\frac{3}{256\alpha^2}t|k-k_2|^3}\leq C\e^{-\frac{1}{1024\alpha^2} t|k+\frac{2}{3}|^3}\leq C\e^{-\frac{1}{6}|z|^3}.
\ee
As a consequence, we have
\be\label{3.48}
\begin{aligned}
&\left|\varrho_a(x,t;k)-\varrho\left(-\frac{2}{3}\right)\right|\e^{-t(\Phi(k)-\Phi(-\frac{2}{3}))}\\
&\leq C|\varrho_a(x,t;k)-\varrho(k_2)|\e^{-t\text{Re}\Phi}
+C\left|\varrho(k_2)-\varrho\left(-\frac{2}{3}\right)\right|\e^{-t\text{Re}\Phi}\\
&\leq C|k-k_2|\e^{-\frac{3}{4}t|\text{Re}\Phi|}
+C\left|k_2+\frac{2}{3}\right|\e^{-t|\text{Re}\Phi|}\leq C\left|zt^{-\frac{1}{3}}\right|\e^{-\frac{1}{6}|z|^3}.
\end{aligned}
\ee
As in the proof of Lemma \ref{lemma3.2}, this implies $\|J^{(1)}-J^r\|_{L^\infty(\Xi_2^\epsilon)}\leq Ct^{-1/3}$, $\|J^{(1)}-J^r\|_{L^1(\Xi_2^\epsilon)}\leq Ct^{-2/3}$.
For $k\in\Xi_3^\epsilon$, we have Re$\Phi=0$, and so, by \eqref{3.43},
\be
\left|\varrho(k)-\varrho\left(-\frac{2}{3}\right)\right|\leq C\left|k+\frac{2}{3}\right|\leq Ct^{-\frac{1}{3}}.
\ee
Thus, Equation \eqref{3.39} holds.
\end{proof}

\subsubsection{The final step}
Let $\breve{\Xi}=\Xi\cup\partial D_\epsilon(-2/3)$, see Figure \ref{fig6}. Define the new matrix-valued function $\breve{M}(x,t;k)$ by
\be
\breve{M}(x,t;k)=\left\{
\begin{aligned}
&M^{(1)}(x,t;k)[M^r(x,t;k)]^{-1},\quad k\in D_\epsilon\left(-\frac{2}{3}\right),\\
&M^{(1)}(x,t;k),\qquad\qquad\qquad\quad{\text elsewhere},
\end{aligned}
\right.
\ee
which satisfies the following $3\times3$ matrix RH problem:\\
$\bullet$ $\breve{M}(x,t;k)$ is analytic outside the contour $\breve{\Xi}$ with continuous boundary values on $\breve{\Xi}$;\\
$\bullet$ For $k\in\breve{\Xi}$,
$\breve{M}_+(x,t;k)=\breve{M}_-(x,t;k)\breve{J}(x,t;k);$\\
$\bullet$ $\breve{M}(x,t;k)\to\mathbb{I}_{3\times3}$, as $k\to\infty$;\\
where the jump matrix $\breve{J}$ is described by
\be
\breve{J}=\left\{
\begin{aligned}
&M^{r}_-J^{(1)}[M^{r}_+]^{-1},\quad k\in\breve{\Xi}\cap D_\epsilon\left(-\frac{2}{3}\right),\\
&[M^{r}]^{-1},\qquad\quad\quad\,\ k\in\partial D_\epsilon\left(-\frac{2}{3}\right),\\
&J^{(1)},\qquad\qquad\qquad  k\in\breve{\Xi}\setminus \overline{D_\epsilon\left(-\frac{2}{3}\right)}.
\end{aligned}
\right.
\ee
Write
\bea
\breve{\Xi}_1=\bfR\setminus[k_1,k_2],\quad
\breve{\Xi}_2=\breve{\Xi}\setminus
\left(\bfR\cup\overline{D_\epsilon\left(-\frac{2}{3}\right)}\right).\nn
\eea
\begin{lemma}
Let $\breve{w}=\breve{J}-\mathbb{I}_{3\times3}$. For $(x,t)\in\mathcal{T}_{\leq}$, we have the following estimates for each $1\leq p\leq\infty$:
\begin{align}
\|\breve{w}\|_{L^p(\partial D_\epsilon(-\frac{2}{3})}\leq&Ct^{-\frac{1}{3}},\\
\|\breve{w}\|_{L^p(\breve{\Xi}_1)}\leq&Ct^{-\frac{3}{2}},\\
\|\breve{w}\|_{L^p(\breve{\Xi}^\epsilon)}\leq& Ct^{-\frac{1}{3}(1+\frac{1}{p})},\\
\|\breve{w}\|_{L^p(\breve{\Xi}_2)}\leq&C\e^{-ct}.
\end{align}
\end{lemma}

The remainder of the derivation in Sector $\mathcal{T}_\leq$ now proceeds as in Sector $\mathcal{T}_\geq$.

\appendix
\section{Coupled Painlev\'e II Riemann--Hilbert problem}\label{secA}

\begin{figure}[htbp]
\centering
\includegraphics[width=3.5in]{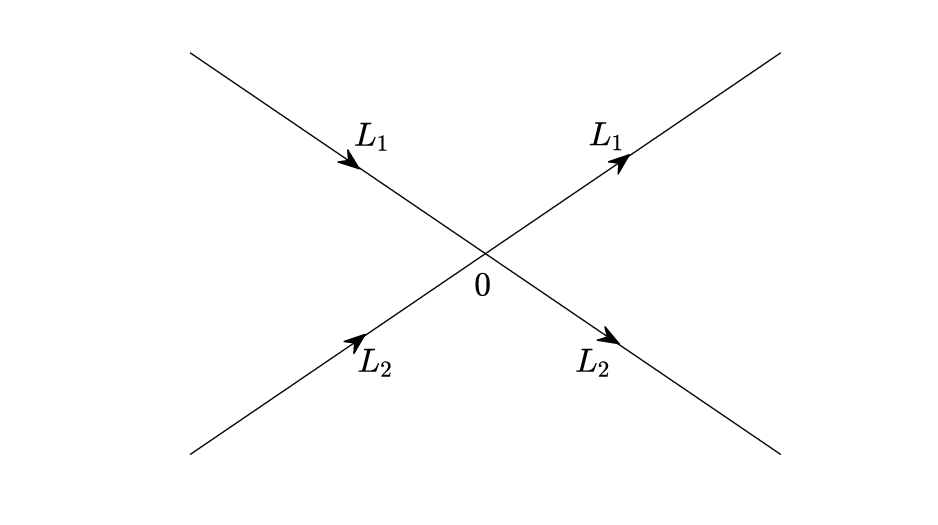}
\caption{The oriented contour $L$.}\label{figL}
\end{figure}
Define the contour $L=L_1\cup L_2$, which is oriented as in Figure \ref{figL}, where
\be\label{A.0}
\begin{aligned}
L_1=&\{z=\epsilon\e^{\frac{\pi\ii}{6}}|\epsilon\geq0\}\cup\{z=\epsilon\e^{\frac{5\pi\ii}{6}}|\epsilon\geq0\},\\
L_2=&\{z=\epsilon\e^{-\frac{\pi\ii}{6}}|\epsilon\geq0\}\cup\{z=\epsilon\e^{-\frac{5\pi\ii}{6}}|\epsilon\geq0\}.
\end{aligned}
\ee
\begin{theorem} [Coupled Painlev\'e II RH problem]\label{lemmaA.1}
Let $s_1$, $s_2$ be two real numbers and define two $3\times3$ matrices $S_1,\ S_2$ by
\be\label{A.1}
S_1=\begin{pmatrix}
1 & 0 & 0 \\
s_1 & 1 & 0 \\
s_2 & 0 & 1
\end{pmatrix},\quad
S_2=\begin{pmatrix}
1 & s_1 & s_2 \\
0 & 1 & 0 \\
0 & 0 & 1
\end{pmatrix}.
\ee
Then the following $3\times3$ matrix RH problem:\\
$\bullet$ Analyticity: $M^L(y;z)$ is analytic in $\bfC\setminus L$ with continuous boundary values on $L$;\\
$\bullet$ Jump condition: For $z\in L$, $M^L_+(y;z)=M^L_-(y;z)J^L(y;z);$\\
$\bullet$ Normalization: $M^L(y;z)\rightarrow \mathbb{I}_{3\times3}$, as $z\rightarrow\infty$;\\
where
\be\label{A.2}
J^L(y;z)=\e^{\frac{2\ii}{3}(\frac{4}{3}z^3+yz)\Sigma}S_n\e^{-\frac{2\ii}{3}(\frac{4}{3}z^3+yz)\Sigma},\quad z\in L_n,\quad n=1,2,
\ee
has a unique solution $M^L(y;z)$ for each $y\in\bfR$. Moreover, there exists smooth functions $\{M_j^L(y)\}$ of $y\in\bfR$ with decay as $y\rightarrow\infty$ such that
\be\label{A.3}
M^L(y;z)=\mathbb{I}_{3\times3}+\sum_{j=1}^N\frac{M_j^L(y)}{z^j}+O(z^{-N-1}), \quad z\rightarrow\infty.
\ee
The $(1,2)$ entry and $(1,3)$ entry of leading coefficient $M_1^L$ are given by
\be\label{A.4}
\left[M_1^L(y)\right]_{12}=\ii u_p(y),\quad \left[M_1^L(y)\right]_{13}=\ii v_p(y),
\ee
where the real-valued functions $u_p(y)$ and $v_p(y)$ satisfy the following coupled Painlev\'e II equation
\be\label{A.5}
\begin{aligned}
u_p^{''}(y)+8\left(u_p^2(y)+v_p^{2}(y)\right)u_p(y)-yu_p(y)=0,\\
v_p^{''}(y)+8\left(u_p^2(y)+v_p^{2}(y)\right)v_p(y)-yv_p(y)=0.
\end{aligned}
\ee
\end{theorem}
\begin{proof}
It first follows from the symmetry
\be\label{A.6}
J^L(y;z)=\left[J^L(y;z^*)\right]^\dag
\ee
and Zhou's vanishing lemma \cite{ZX} that the solution $M^L(y;z)$ of RH problem exists and is unique, and has an expansion of the form \eqref{A.3}. On the other hand, the coefficients $M_j^L$ and their $y$-derivatives have exponential decay as $y\rightarrow\infty$ by a A Deift--Zhou steepest descent analysis.

Now we put $\Psi(y;z)=M^L(y;z)\e^{\frac{2\ii}{3}(\frac{4}{3}z^3+yz)\Sigma}$. Then the matrix-valued function $\mathcal{U}(y;z)$ defined by
\be\label{A.7}
\mathcal{U}\doteq\Psi_y\Psi^{-1}=\left(M_y^L+\frac{2\ii}{3} zM^L\Sigma\right)\left[M^L\right]^{-1}
\ee
is an entire function of $z$. Thus, we immediately find that $\mathcal{U}(y;z)=\mathcal{U}_1(y)z+\mathcal{U}_0(y)$, and hence \eqref{A.7} becomes
\be\label{A.8}
M_y^L+\frac{2\ii}{3} zM^L\Sigma=(\mathcal{U}_1z+\mathcal{U}_0)M^L.
\ee
Substituting the expansion \eqref{A.3} into \eqref{A.8}, the result implies that
\be\label{A.9}
\mathcal{U}_1=\frac{2\ii}{3}\Sigma,\quad \mathcal{U}_0=-\frac{2\ii}{3}[\Sigma,M_1^L].
\ee
Analogously, we construct another matrix-valued function $\mathcal{V}(y;z)$ as follows:
\be\label{A.10}
\mathcal{V}\doteq\Psi_z\Psi^{-1}
=\left(M_z^L+\frac{2\ii}{3}(4z^2+y)M^L\Sigma\right)\left[M^L\right]^{-1},
\ee
which also is entire, and hence, one can get
\be\label{A.11}
\mathcal{V}=\mathcal{V}_2z^2+\mathcal{V}_1z+\mathcal{V}_0=\frac{8\ii}{3}\Sigma z^2+4\mathcal{U}_0z-\frac{8\ii}{3}[\Sigma,M_2^L]-4\mathcal{U}_0M_{1}^L+\frac{2\ii}{3}y\Sigma.
\ee
However, by substituting the Equation \eqref{A.3} into \eqref{A.8} and collecting term with $O(z^{-1})$, we find
\be\label{A.12}
M_{1y}^L-\frac{2\ii}{3}[\Sigma,M_2^L]=\mathcal{U}_0M_1^L,
\ee
which implies that
\be
\mathcal{V}_0=-4M^L_{1y}+\frac{2\ii}{3}y\Sigma.
\ee
In fact, we have amazedly shown that $\Psi$ admits a Lax pair
\be\label{A.14}
\left\{
\begin{aligned}
\Psi_y=&\mathcal{U}\Psi,\\
\Psi_z=&\mathcal{V}\Psi,
\end{aligned}
\right.
\ee
where $\mathcal{U}$ and $\mathcal{V}$ are expressed in terms of $M_1^L(y)$. Thus, the compatibility condition
\be
\mathcal{U}_z-\mathcal{V}_y+\mathcal{U}\mathcal{V}-\mathcal{V}\mathcal{U}=0
\ee
of the Lax pair \eqref{A.14} gives that
\be\label{A.16}
4M^Y_{1yy}+\frac{8\ii}{3}[\Sigma,M_1^L]M_{1y}^L-\frac{8\ii}{3} M_{1y}^L[\Sigma,M_1^L]+\frac{4}{9}y[\Sigma,M_1^L]\Sigma-\frac{4}{9}y\Sigma[\Sigma,M_1^L]=0.
\ee

Moreover, note that the jump matrix $J^L(y;z)$ obeys another symmetry
\be
J^L(y;z)=\left[J^L(y;-z^*)\right]^*,
\ee
which together with the relation \eqref{A.6} imply that $M^L(y;z)$ satisfies
\be\label{A.18}
M^L(y;z)=\left[\left(M^L\right)^\dag(y;z^*)\right]^{-1}
=\left[M^L(y;-z^*)\right]^*.
\ee
Therefore, the leading-order coefficient $M_1^L(y)$ satisfies the symmetries
\be\label{A.19}
M_1^L(y)=-\left[M_1^L(y)\right]^\dag=-\left[M_1^L(y)\right]^*.
\ee
As a consequence, we can write
\be\label{A.20}
M_1^L(y)=\begin{pmatrix}
f_1(y) & f_2(y) & f_3(y) \\
f_2(y) & f_4(y) & f_5(y) \\
f_3(y) & f_5(y) & f_6(y)
\end{pmatrix},
\ee
where $f_j(y)\in\ii\bfR$ for $j=1,\cdots,6$.

In the following, substituting \eqref{A.20} respectively into \eqref{A.16} and \eqref{A.12}, we obtain that
\be
\begin{aligned}
&f''_1-4\ii\left(f_2f'_2+f_3f'_3\right)=0,\\
&f''_2-2\ii\left(f_2f'_4+f_3f'_5-f_2f'_1\right)-yf_2=0,\\
&f''_3-2\ii\left(f_2f'_5+f_3f'_6-f_3f'_1\right)-yf_3=0,\\
&f'_4+2\ii f_2^2=0,\ f'_5+2\ii f_2f_3=0,\ f'_6+2\ii f_3^2=0.
\end{aligned}
\ee
In addition, since $M^L_1(y)$ and its derivatives decay as $y\to\infty$, we quickly find that
\be
f'_1=2\ii(f_2^2+f_3^2).
\ee
Thus, we obtain
\be\label{A.23}
\begin{aligned}
f''_2(y)-8\left(f_2^2(y)+f^2_3(y)\right)f_2(y)-yf_2(y)=0,\\
f''_3(y)-8\left(f_2^2(y)+f^2_3(y)\right)f_3(y)-yf_3(y)=0.
\end{aligned}
\ee
Then the lemma follows by setting $f_2(y)=\ii u_p(y)$ and $f_3(y)=\ii v_p(y)$.
\end{proof}
\section{A model Riemann--Hilbert problem in sector $\mathcal{T}_\leq$}\label{secb}

\begin{figure}[htbp]
\centering
\includegraphics[width=4in]{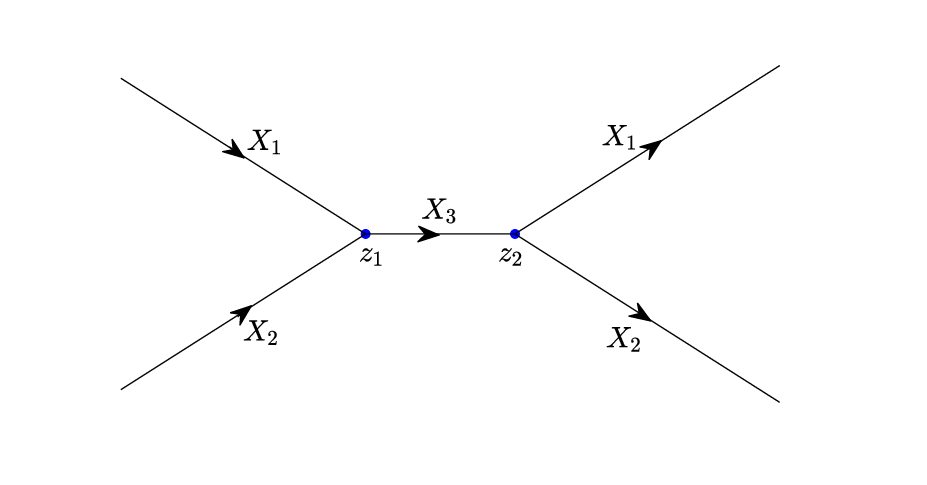}
\caption{The oriented contour $X$.}\label{fig8}
\end{figure}
Given two real numbers $z_2\geq0$ and $z_1=-z_2$, let $X$ denote the contour $X=\bigcup_{j=1}^3X_j$ shown in Figure \ref{fig8}, where the line segments
\be\label{B.1}
\begin{aligned}
X_1&=\{z=z_2+l\e^{\frac{\pi\ii}{6}}|l\geq0\}\cup\{z=z_1+l\e^{\frac{5\pi\ii}{6}}|l\geq0\},\\
X_2&=\{z=z_2+l\e^{-\frac{\pi\ii}{6}}|l\geq0\}\cup\{z=z_1+l\e^{-\frac{5\pi\ii}{6}}|l\geq0\},\\
X_3&=\{z=l|z_1\leq l\leq z_2\}.
\end{aligned}
\ee
Define a $1\times2$ row-vector $p$ by
\be\label{B.2}
p=\begin{pmatrix}s_1 & s_2\end{pmatrix},\quad s_1,s_2\in\bfR.
\ee
We consider the following family of $3\times3$ matrix RH problems parametrized by $y\leq0$, $z_2\geq0$:\\
$\bullet$ $M^X(y,z_2;z)$ is analytic in $\bfC\setminus X$ with continuous boundary values on $X$;\\
$\bullet$ For $z\in X$, $M^X_+(y,z_2;z)=M^X_-(y,z_2;z)J^X(y,z_2;z);$\\
$\bullet$ $M^X(y,z_2;z)\rightarrow \mathbb{I}_{3\times3}$ as $z\rightarrow\infty$;\\
where the jump matrix $J^X(y,z_2;z)$ is defined by
\be\label{B.3}
J^X(y,z_2;z)=\left\{
\begin{aligned}
&\begin{pmatrix}
1 & \mathbf{0}_{1\times2} \\
p^\dag\e^{2\ii(yz+\frac{4}{3}z^3)} & \mathbb{I}_{2\times2}
\end{pmatrix},\qquad\qquad\qquad\qquad\qquad\quad z\in X_1,\\
&\begin{pmatrix}
1 & p\e^{-2\ii(yz+\frac{4}{3}z^3)}\\
\mathbf{0}_{2\times1} & \mathbb{I}_{2\times2}
\end{pmatrix},\qquad\qquad\qquad\qquad\quad\qquad z\in X_2,\\
&\begin{pmatrix}
1 & p\e^{-2\ii(yz+\frac{4}{3}z^3)}  \\
\mathbf{0}_{2\times1} & \mathbb{I}_{2\times2}
\end{pmatrix}\begin{pmatrix}
1 & \mathbf{0}_{1\times2} \\
p^\dag\e^{2\ii(yz+\frac{4}{3}z^3)} & \mathbb{I}_{2\times2}
\end{pmatrix},\quad z\in X_3.
\end{aligned}
\right.
\ee
Define the parameter subset
\be\label{B.4}
\Bbb P=\{(y,z_2)\in\bfR^2|-C_1\leq y\leq0,\ \sqrt{-y}/2\leq z_2\leq C_2\},
\ee
where $C_1$, $C_2$ are two positive constants.
\begin{theorem}\label{lemmaB.1}
Let $p$ be of the form \eqref{B.2}. The RH problem for $M^X$ with jump matrix $J^X$ given by \eqref{B.2} has a unique solution $M^X(y,z_2;z)$ whenever $(y,z_2)\in\Bbb P$.
 There are smooth functions $\{M_j^X(y)\}$ such that
\be\label{B.5}
M^X(y,z_2;z)=\mathbb{I}_{3\times3}+\sum_{j=1}^N\frac{M_j^X(y)}{z^j}+O(z^{-N-1}), \quad z\rightarrow\infty.
\ee
The $(1,2)$ entry and $(1,3)$ entry of leading coefficient $M_1^X$ are given by
\be\label{B.6}
\left[M_1^X(y)\right]_{12}=\ii u_p(y),\quad \left[M_1^X(y)\right]_{13}=\ii v_p(y),
\ee
where the real-valued functions $u_p(y)$ and $v_p(y)$ are the smooth solution of the coupled Painlev\'e II equation \eqref{A.5} associated with $s_1$, $s_2$ according to Theorem \ref{lemmaA.1}. Furthermore, $M^X(y,z_2;z)$ is uniformly bounded for $z\in\bfC\setminus Z$.
\end{theorem}
\begin{proof}
It follows the same line as Theorem 5.1 in \cite{L-SIAM} or Lemma B.1 in \cite{LG-JDE}.
\end{proof}

\medskip
\small{

}
\end{document}